\documentclass[oneside,11pt]{article}

\usepackage{amsmath}
\usepackage{amssymb}
\usepackage{pxfonts}
\usepackage{dsfont}
\usepackage{graphicx}
\usepackage{eucal}
\usepackage{mathrsfs}
\usepackage{theorem}
\usepackage{pifont}
\usepackage{sectsty}
\usepackage{amscd}
\usepackage{color}
\usepackage{fancyhdr}
\usepackage{framed}
\usepackage[all]{xy}
\usepackage[backref=page]{hyperref}

\definecolor{shadecolor}{rgb}{0.8,0.8,0.8}

\theoremheaderfont{\fontfamily{pzc}\bfseries\large}
\newtheorem{theorem}{Theorem}[section]

\newtheorem{proposition}[theorem]{Proposition}

\newtheorem{corollary}[theorem]{Corollary}
\newtheorem{definition}[theorem]{Definition}

\newtheorem{claim}[theorem]{Claim}
\newtheorem{example}[theorem]{Example}

\theorembodyfont{\rmfamily}

\newcommand{\specexercise}[1]{}

\newenvironment{proof}{{\flushleft \emph{Proof}:}}{\hfill\ding{110}}

\newenvironment{remark}{{\flushleft \fontfamily{pzc}\bfseries\large Remark:}}{}

\newcommand{\Emph}[1]{{\slshape\bfseries #1}}  

\newcommand{\g}{G}

\newcommand{\Vol}{\text{Vol}}

\newcommand{\M}{\mathcal{M}}

\newcommand{\R}{\mathbb{R}}

\newcommand{\vp}{\varphi}
\newcommand{\dist}{\operatorname{dist}}

\newcommand{\SO}{\operatorname{SO}}
\newcommand{\GL}{\operatorname{GL}}
\newcommand{\id}{\operatorname{Id}}
\renewcommand{\det}{\operatorname{det}}
\renewcommand{\P}{\mathcal{P}}

\newcommand{\e}{\varepsilon}
\newcommand{\W}{\Omega}

\newcommand{\calW}{\mathcal{W}}

\newcommand{\beq}{\begin{equation}}
\newcommand{\eeq}{\end{equation}}

\newcommand{\brk}[1]{\left(#1\right)}          

\numberwithin{equation}{section}

\begin{document}

\title{On material-uniform elastic bodies with disclinations and their homogenization}
\author{Cy Maor\footnote{Einstein institute of Mathematics, Hebrew University of Jerusalem.}}
\date{}
\maketitle

{\centering\footnotesize \emph{Dedicated to Marcelo Epstein, in celebration of his 80th birthday.}\par}

\begin{abstract}
In this note, we define material-uniform hyperelastic bodies (in the sense of Noll) containing discrete disclinations and dislocations, and study their properties.
We show in a rigorous way that the size of a disclination is limited by the symmetries of the constitutive relation; in particular, if the symmetry group of the body is discrete, it cannot admit arbitrarily small, yet non-zero, disclinations.
We then discuss the application of these observations to the derivations of models of bodies with continuously-distributed defects.
\end{abstract}

\section{Introduction}
The systematic study of material defects has a long history, dating back to Volterra \cite{Vol07} at the beginning of the 20th century.
Beginning in the mid-20th century, two somewhat different paradigms for describing bodies with defects (either discrete or continuously-distributed) have emerged in the rational mechanics literature.
In the first one, initiated by Nye \cite{Nye53}, Kondo \cite{Kon55}, and Bilby \cite{BBS55}, the elemental object is the body manifold, and its associated geometric fields --- typically a Riemannian metric and an affine connection --- describe the existence of defects in it.
In this approach, the mechanical behavior (i.e., the constitutive relation) should relate to these fields, although the exact  relationship  typically remains quite vague.
In the second paradigm, mainly due to Noll \cite{Nol58} and Wang \cite{Wan67}, the constitutive relation is the elemental object, from which one can associate various geometric fields, under the assumption of material-uniformity (meaning, in a precise sense, that distinct material points have ``the same'' mechanical response).
See \cite{EKM20} for a recent account of the differences and similarities between the approaches.

In this geometric paradigm, a body with continuously-distributed defects is often modeled by a triplet $(\M,\g,\nabla)$, where $(\M,\g)$ is a Riemannian manifold and $\nabla$ is an affine connection.
The connection $\nabla$ is uniquely characterized by three tensor fields: its curvature, torsion and non-metricity tensors (the first two are intrinsic to $\nabla$, the last one is with respect to $\g$).
These fields correspond, in this paradigm, to fields of disclinations, dislocations, and point-defects, respectively (see \cite{MR02,YG12,YG12b,YG13,RG17} for some rather recent accounts).
When specific constitutive relations are considered in these works, they are often neo-Hookean or other isotropic models (e.g., \cite{YG12,YG13}. See \cite{GY18} for examples of non-isotropic models).
It is worth noting, however, that crystalline materials --- arguably the most common examples of defected bodies --- often have discrete symmetry groups (also called isotropy groups) due to their crystalline structure.

Although these models for continuously-distributed defects are well known, the homogenization problem --- that is, the rigorous derivation of a model for continuously-distributed defects from models with discretely distributed ones --- was largely unexplored until the last decade.
For continuously distributed dislocations, this problem was addressed in a series of works that first considered only the geometry of the body manifold \cite{KM15,KM16} and later included the homogenization of its mechanical response \cite{KM16a,EKM20} (see also \cite{KO20} for a different approach to geometric homogenization).
A notable outcome of these homogenization results concerns the symmetry group of the material. Specifically, if the material has a discrete symmetry group, the limiting geometric fields $\g$ and $\nabla$ are uniquely derived from the constitutive relation (i.e., there is a one-to-one correspondence between the paradigms). 
In contrast, if the material is isotropic, the emerging constitutive relation is completely independent of the torsion tensor of $\nabla$, which represents the distribution of dislocations in the geometric paradigm \cite{EKM20}. 
In this case, while the paradigms are consistent, the distribution of dislocations cannot be inferred from the mechanical behavior alone.

The aim of this note is to investigate a framework for similar problems for bodies with disclinations.
Disclinations are both simpler and more complicated than dislocations.
On one hand, their geometry is simpler: it is easy to visualize a body with a disclination (a cone) and to see how such bodies can approximate a smooth surface (in the same sense that a soccer ball approximates a sphere). 
On the other hand, disclinations are more complicated: while a body with dislocations admits global parallelism (i.e., ``lattice directions" can be defined uniquely across the entire body), a body with finitely-many disclinations only admits local parallelism.
Consequently, describing a non-isotropic constitutive relation for such bodies is more challenging.

\paragraph{Main results and structure of the paper.}
The basic modeling assumption of this work is that the elastic body is a continuum (a smooth manifold) that is material-uniform. 
For simplicity, we assume that its constitutive relation is hyperelastic, but impose no additional assumptions (e.g., isotropy).\footnote{One could argue that a more elemental model to start with, in particular in considering crystalline defects, is a discrete model of particles. However discrete-to-continuum limit, unlike homogenization of defects, is quite well-understood; see, e.g., \cite{CvM21}, for a rigorous derivation of a discrete-to-continuum limit in the presence of a disclination.}

In \S\ref{sec:definitions}, we provide a brief overview of hyperelastic material-uniform bodies, including their symmetry groups and the concept of a material connection.
In \S\ref{sec:defects}, we focus on bodies with discrete defects, and show the following:
\begin{enumerate}
\item We present a general definition of a hyperelastic material-uniform body with (discrete) defects.
\item We prove that such bodies admit a locally-flat material connection, which induces a notion of parallelism on any sub-body that does not contain a disclination.
\item We revisit how the disclination and dislocation contents are derived from the material connection.
\item We prove a relationship between the disclination content and the symmetry group of the material.
\end{enumerate}
Finally, in \S\ref{sec:discussion}, we discuss the implications of these results for the homogenization of defects. Specifically:
\begin{itemize}
\item We show that, unlike dislocations, disclinations cannot be homogenized when the symmetry group is discrete, as disclinations cannot be arbitrarily small.
\item For isotropic materials, it is possible to homogenize disclinations, but the limiting constitutive relation depends solely on the limiting Riemannian metric, irrespective of the type of defects from which it was derived.
\end{itemize}
We then revisit the geometric paradigm of continuous distributions of defects, in which defect fields are associated with the tensors of a material connection. 
The observations above reveal the limitations on extending this paradigm --- beyond the theory of dislocation --- to derive effective homogenization models that account for mechanical behavior.

\paragraph{Acknowledgements}
This work is dedicated to Marcelo Epstein, who has immensely shaped my academic journey and serves as an inspiration both as an academic and as a person. 
It is a pleasure to present this paper, directly influenced by our discussions, in honor of his 80th birthday.
The author is grateful to Raz Kupferman for helpful discussions related to this manuscript.
This work was written while the author was visiting the University of Toronto and the Fields Institute for Research in Mathematical Sciences; their support and hospitality is gratefully acknowledged.
The author was partially supported by ISF grant 2304/24 and BSF grant 2022076.

\section{Material-uniform hyperelastic bodies}\label{sec:definitions}

We start by defining a general hyperelastic body: a body whose mechanical behavior can be obtained by minimizing an elastic energy.
\begin{definition}[Hyperelastic body]
\label{def:elastic_body}
A \Emph{hyperelastic body} is a triplet $(\M,W,\Vol)$, consists of a $d$-dimensional oriented differentiable manifold, $\M$---the \Emph{body manifold}---a volume form $\Vol$, and an energy-density function (or \Emph{constitutive relation}),
\[
W: T^*\M\otimes\R^d \to (-\infty,\infty],
\]
which is viewed as a (nonlinear) bundle map over $\M$.\footnote{Allowing for $W$ to obtain the value $+\infty$ is a standard way for excluding some non-physical maps, like orientation-reversing ones. For viewing $W$ as a true bundle map one can think of $ (-\infty,\infty]$ as the trivial bundle $ (-\infty,\infty] \times \M$.}
The \Emph{elastic energy} associated with a configuration $f:\M\to \R^d$ is given by
\beq\label{eq:energy}
E_{(\M,W,\Vol)}(f) = \int_\M W_p(df(p))\, \Vol(p).
\eeq
\end{definition}
While it is not necessary from the definition, we typically (and in this paper) assume that $\M$ can be covered in a single chart.\footnote{Note that in \cite{EKM20} the volume-form is not part of the definition of the hyperelastic body, yet it is required implicitly, when writing the elastic energy.}

In what follows we will need the notion of $\R^d$-valued forms: 
An $\R^d$-valued $k$-form over a manifold $\M$ is a section of $\bigwedge^k T^*\M\otimes \R^d$.
For example, if $f:\M\to \R^d$, then its derivative $df$ is an $\R^d$-valued one form.
$\R^d$-valued $k$-forms can be considered as $d$-tuples of (standard) $k$-forms, and as such, all the standard operations on forms extend to them component-wise: an exterior derivative of an $\R^d$-valued $k$-form is an $\R^d$-valued $(k+1)$-form, the form is closed if it exterior derivative vanishes, and it is exact if it is an exterior derivative of an $\R^d$-valued $(k-1)$-form.
Similarly, one can integrate an $\R^d$-valued $1$-form over a curve (obtaining a vector in $\R^d$), and so on.

Our focus is on bodies in which ``all the material points behave the same way''.
This property is called \Emph{material-uniformity} and was developed by Noll \cite{Nol58} and Wang \cite{Wan67}.
For hyperelastic bodies it reads as follows (see also \cite{EKM20}):\footnote{In \cite{EKM20} the definition is via the maps $E=\P^{-1}$, which can be thought of a choice of a frame at each tangent space. Here we use $\P$ since it is more convenient for our purposes to use one-forms. Note also that in \cite{EKM20} the volume-preserving requirement is only implicit.}
\begin{definition}[Hyperelastic smooth uniform body]
A \Emph{hyperelastic smooth uniform body} is a hyperelastic body $(\M,W,\Vol)$ that can be represented as a triplet $(\M,\{U_\alpha, \P_\alpha\}_{\alpha},\calW)$, where $(U_\alpha)_{\alpha}$ is an open cover of $\M$, and $\P_\alpha:TU_\alpha \to \R^d$ are $\R^d$-valued one-forms satisfying:
\begin{itemize}
\item For any point $p\in U_\alpha$, $\P_\alpha(p):T_p\M \to \R^d$ is a volume-preserving isomorphism with respect to the volume form $\Vol(p)$ on $T_p\M$ and the Euclidean volume form $dx$ on $\R^d$ (that is, $\Vol = (\P_\alpha)^*dx$).
\item $\calW:\R^d\otimes \R^d\to (-\infty,\infty]$ is related to $W$ via $W_p = (\P_{\alpha}(p))_*\calW$, that is,
	\[
	W_p(A_p) = \calW(A_p\circ \P_\alpha^{-1}(p)),\qquad \forall p\in U_\alpha, \,\, \forall A_p\in T_p^*M\otimes \R^d.
	\]
\end{itemize}
\end{definition}

The function $\calW$ represents the mechanical behavior of the ``particle'' from which the body is composed, and thus it is often called an \Emph{archtype}. 
The pair $(U_\alpha,\P_\alpha)$ is called \Emph{reference chart}, $U_\alpha$ a \Emph{reference neighborhood} and $\P_\alpha$ a \Emph{reference map} \cite[Definition~2.7]{Wan67}.
The reference map $\P_\alpha(p)$ represents how the archtype $\calW$ is implanted in the material point $p$; its inverse $\P_\alpha^{-1}$ is therefore sometimes called an \Emph{implant map}.
$\{\P_\alpha\}_\alpha$ also induce a Riemannian metric $G$ on $\M$, via 
\beq\label{eq:intrinsic_metric}
G_p(X,Y) = \P_\alpha(p)(X)\cdot \P_\alpha(p)(Y), \qquad p \in U_\alpha, \,\, X,Y\in T_p\M,
\eeq
where $\cdot$ is the Euclidean inner-product on $\R^d$.
This metric is well-defined (independent of $\alpha$) for solid bodies \cite[Definition~6]{EKM20}.
In the context of non-Euclidean elasticity, where $G$ is typically an elemental object, $\P_\alpha$ are often called prestrain maps.
In other contexts, the map $\P_\alpha$ may also play the role of the plastic strain in Kr\"oner's decomposition  (or, more historically-accurately, Bilby--Kr\"oner--Lee decomposition \cite{SY17}), hence the use of the letter $\P$. See \cite[\S2]{KM23} for further discussion.

Let $p\in U_\alpha$, and let $\{e_i\}_{i=1}^d$ be the standard basis of $\R^d$. 
We can think of $(\P_{\alpha}^{-1}(p)[e_i])_{i=1}^d$ as being lattice directions at the point $p$.
Motivated by this, we call the collection $(U_\alpha,\P_\alpha)_\alpha$ the \Emph{lattice structure} of $\M$.

The representation of $(\M,W,\Vol)$ as $(\M,\{U_\alpha, \P_\alpha\}_{\alpha},\calW)$ is certainly non-unique.
Note that our assumption that $\M$ can be covered by a single chart does not imply that a single, globally defined smooth map $\P:T\M\to \R^d$ that represents the elastic behavior exists. 
As we shall see, the presence of a disclination prevents this possibility.

\begin{definition}[symmetry group]
\label{def:isotropy}
Let $(\M,W,\Vol)$ be a uniform, smooth, hyperelastic body. 
The \Emph{symmetry group} of the body associated with an archetype $\calW$ of $(\M,W,\Vol)$ is a group $\mathcal{G} \le \GL(d)$, defined by
\[
\calW(B\circ g) = \calW(B) \quad \text{for every $B\in \R^d\otimes \R^d$ and $g\in \mathcal{G}$}.
\]
The body is called a \Emph{solid} if there exists a $\calW$ such that $\mathcal{G} \le \SO(d)$.
We shall only consider solid bodies, and only consider such $\calW$, called $\calW$ \Emph{undistorted}, as admissible.
\end{definition}

\begin{example}\label{ex:archetypes}
\begin{enumerate}
\item \Emph{Isotropic materials} are ones for which the symmetry group is exactly $\SO(d)$. 
	For example, neo-Hookean archetypes such as 
	\[\calW(B) = \|B\|^2 + (\det B -1)^2 \quad \text{or} \quad \calW(B) = \dist^2(B, \SO(d))\]
	 are isotropic.
\item An example for a two-dimensional archetype with a discrete symmetry group is 
	\[\calW(B) = \sum_{k=0}^{n-1}(|Bv_k| - 1)^2 + (\det B - 1)^2\]
	for some $n\in \mathbb{N}$ and $v_k = (\cos (k\pi/2n), \sin (k\pi/2n))$. For $n=2$ we obtain a square symmetry, and for $n=3$ a hexagonal symmetry.
\end{enumerate}
Note that for $d=2$ a solid is either isotropic or has discrete symmetry group.
\end{example}

Note that if the body can be represented as $(\M,\{U_\alpha, \P_\alpha\}_{\alpha},\calW)$, and $p\in U_\alpha\cap U_\beta$ for some $\alpha,\beta$, then 
\beq\label{eq:symmetry_group}
\P_{\alpha}(p)\circ \P_\beta^{-1}(p) \in \mathcal{G}.
\eeq
Indeed, for every $B\in \R^d\otimes \R^d$, we have
\[
\calW(B \circ \P_{\alpha}\circ \P_\beta^{-1}(p)) = W_p(B \circ \P_{\alpha}(p)) = \calW(B \circ \P_{\alpha}\circ \P_{\alpha}^{-1}(p)) = \calW(B).
\] 
The relation \eqref{eq:symmetry_group} implies that the lattice directions at a point are only defined up to the symmetry group of the body.

Finally, we recall the definition of a material connection, which is the global object that represents how ``lattice directions" (i.e., tangent vectors) at different material points correspond to each other.
In the geometric paradigm, the connection is the elemental object in the description of a solid.
\begin{definition}[Material connection]
\label{def:material_connection}
A \Emph{material connection} of $(\M,W)$ is an affine connection $\nabla$ on $\M$ whose parallel-transport operator $\Pi$ leaves $W$ invariant.
That is, for every $p,q\in \M$, $A\in T^*_q\M\otimes \R^d$ and path $\gamma$ from $p$ to $q$,
\beq\label{eq:def_mat_connection}
W_p(A\circ \Pi_\gamma) = W_q(A),
\eeq
where $\Pi_\gamma:T_p\M\to T_q\M$ is the parallel transport along $\gamma$.
\end{definition}

In general, a material connection may fail to exist, or may not be unique. 
In fact, for isotropic materials, any connection that is metrically-consistent with the metric $G$ defined in \eqref{eq:intrinsic_metric} is a material connection.
However, if the isotropy group is discrete, it admits a unique connection, which is locally-flat (has zero Riemannian curvature; it can still have a non-trivial holonomy) \cite[Proposition 1]{EKM20}.
That is, at least locally, the body admits local parallelism: a frame field at the vicinity of each material point that is parallel (with respect to the material connection), and represents how lattice directions change in the body.
As we will see in Proposition~\ref{prop:material_connection} below, the existence of a locally-flat material connection extends to bodies with disclinations and dislocations.

\section{Bodies with discrete defects}\label{sec:defects}
A body is \Emph{Euclidean} if we can consider the body manifold as a subset $\W$ of $\R^d$ and write its associated energy as
\[
E(f) = \int_\W \calW(df)\,dx.
\]
That is, if we can represent it as $(\W,\{\W,\id\},\calW)$ for an undistorted $\calW$.
Note that if $F: \M \to \W$ is a diffeomorphism, we can change variables and obtain,
\[
E(f) = \int_\M \calW(df' \circ dF^{-1})\, F^*dx
\]
where $f' = f\circ F$, and $F^*dx$ is the pullback of $dx$.
Thus, $(\M,\{\M,dF\},\calW)$ is also a representation of the same body.

Bodies with discrete defects are multiply-connected bodies that locally (i.e., everywhere except for the defects' loci) look like Euclidean bodies.
By the above discussion, we obtain that it is a body whose reference maps are locally differentials --- exact forms.
Since a locally-exact form is a closed form, this amounts to requiring that the reference maps are closed forms:
\begin{definition}\label{def:defects}
\begin{enumerate}
\item \Emph{A body with discrete defects} (disclinations and dislocations) is a uniform, smooth, hyperelastic body $(\M,\{U_\alpha, \P_\alpha\}_{\alpha},\calW)$, such that for any $\alpha$, $\P_\alpha$ is closed, i.e., $d\P_\alpha = 0$.

\item A sub-body $\M'\subset \M$ has \Emph{no disclinations} if one can add to the lattice structure a reference neighborhood $(\M',\P')$ with $\P'$ being closed.

\item A sub-body $\M'\subset \M$ has \Emph{no defects} if  one can add to the lattice structure a reference neighborhood $(\M',\P')$ with $\P'$ being exact, i.e., $P' = dF$, for some $F:\M'\to \R^d$.
In particular, every simply-connected sub-body contains no defects.
\end{enumerate}
\end{definition}

We will motivate these definitions below.
Note that Noll \cite{Nol58} assumes that the body can be covered by a single reference neighborhood, which is not sufficient for admitting disclinations; thus we need the more general definition of a body due to Wang \cite{Wan67}.

\begin{proposition}[Existence of material connection]
\label{prop:material_connection}
The Levi-Civita connection of the Riemannian metric \eqref{eq:intrinsic_metric} associated with a body with discrete defects is a material connection. It is locally-flat, and has a trivial holonomy on every sub-body with no disclinations.
\end{proposition}

\begin{proof}
Denote by $\g$ be the Riemannian metric associated with the body, and let $p\in U_\alpha$.
Then, for some neighborhood $V\subset U_\alpha$ of $p$, we have that $\P_\alpha|_V$ is exact, hence $\P_{\alpha}|_{V} = dF$ for some $F : V \to \R^d$.
Thus, $\g|_V$ is simply the pullback of the Euclidean metric to $V$ via $F$, and its Levi-Civita connection is the pullback of the Euclidean one; in particular, it has zero curvature, that is, it is locally-flat.
We now claim that for any curve $\gamma\subset U_\alpha$ between points $p$ and $q$, the parallel transport of this connection is given by
\beq\label{eq:Pi_gamma}
\Pi_\gamma = \P_\alpha(q)^{-1}\P_\alpha(p).
\eeq
Indeed, we can partition $\gamma$ as $\gamma = \gamma_n *\ldots *\gamma_1$, where each $\gamma_i$ is contained in a simply-connected open set $V_i\subset U_\alpha$. 
We then have
\beq\label{eq:concatenation}
\Pi_\gamma = \Pi_{\gamma_n} \circ \ldots \circ \Pi_{\gamma_1}.
\eeq
Thus it is sufficient to prove \eqref{eq:Pi_gamma} for each $\gamma_i$, but this just follows from the fact that $\P_\alpha|_{V_i}=dF$ is exact and thus the parallel transport is the pullback of the (trivial) Euclidean parallel transport via $F$.
Formula \eqref{eq:Pi_gamma} now shows that for any closed curve $\gamma\subset U_\alpha$ based in $p$, the parallel transport satisfy $\Pi_\gamma = \id_{T_p\M}$, that is, the holonomy of $\gamma$ is trivial.
By definition, this extends to all closed curves in sub-bodies that contain no disclinations.

It remains to show that $\nabla$ is a material connection.
It is immediate from \eqref{eq:Pi_gamma} that \eqref{eq:def_mat_connection} holds for any $\gamma\subset U_\alpha$.
By partitioning a curve $\gamma$ into $\gamma_i \subset U_{\alpha_i}$ and using \eqref{eq:concatenation}, it follows in full generality.
\end{proof}

The Levi-Civita connection of a body with discrete defects enables us to quantify the disclination and dislocation content in the body:
\begin{definition}[disclination content and Burgers vectors]\label{def:defects_content}
Let $(\M,\{U_\alpha, \P_\alpha\}_{\alpha},\calW)$ be a body with discrete defects. 
Denote its associated Levi-Civita connection by $\nabla$, and its parallel transport operator along a curve $\gamma:[s,t]\to \M$ by $\Pi_{\gamma(s)}^{\gamma(t)}$.
Let $\gamma:[0,1]\to \M$ be a piecewise smooth closed curved, based at $p\in \M$.
\begin{enumerate}
\item The \Emph{disclination content} associated with $\gamma$ is the map
	\beq\label{eq:disclination_content}
	c_\gamma :=\Pi_{\gamma(0)}^{\gamma(1)} \in T_p^*\M\otimes T_p\M,
	\eeq
	that is, the holonomy associated with the curve $\gamma$.
	We say that $\gamma$ has zero disclination content if its disclination content is $\id_{T_p\M}$.
\item If $\gamma$ has zero disclination content, then its associated \Emph{Burgers vector} (dislocation content) is the vector
	\beq\label{eq:Burgers}
	b_\gamma := \int_0^1 \brk{\Pi_{\gamma(0)}^{\gamma(t)}}^{-1}\dot\gamma(t)\,dt \in T_p\M.
	\eeq
\end{enumerate}
\end{definition}
Since the holonomy of $\nabla$ is locally trivial, it follows that the disclination content of a curve is homotopy-invariant, and it is trivial if the curve is contained in a sub-body with no disclinations.
The Burgers vector of a curve is also homotopy-invariant among curves that have zero disclination content. 
The formula \eqref{eq:Burgers} makes sense for general bodies with discrete defects, however it is then not a homotopy-invariant --- indeed, this is a manifestation of the well-known fact that in the presence of disclinations, the Burgers vector around a defect core depends on the chosen Burgers circuit (see also \cite{KMS15}).

\begin{remark}
\begin{enumerate}
\item It follows from Proposition~\ref{prop:material_connection} that the disclination content of any closed curve $\gamma$ that is contained in a sub-body with no disclinations (according to Definition~\ref{def:defects}) is trivial.
	It thus follows that a body that contains disclinations cannot be represented by a single (smooth) reference chart.
\item Furthermore, for a sub-body containing no defects (according to Definition~\ref{def:defects}), the Burgers vector of any closed curve contained in it is zero.
	Indeed, in this case the connection $\nabla$ on the sub-body $\M'$ is simply the pullback of the Euclidean connection on $\R^d$ by an immersion $F:\M'\to \R^d$.
\end{enumerate}
This discussion motivates the nomenclature of Definition~\ref{def:defects}.
\end{remark}

\begin{remark}
There are cases which, according to Definition~\ref{def:defects}, the body contains no defects, yet, one might consider them as having defects: 
for example, if $\M\subset \R^2$ is an annulus, and $\P$ is a global reference map that is the derivative of a double cover map $f:\M\to\M$. 
Even though the holonomy in this case is trivial, one can think of this body as having a negative disclination of excess angle $-2\pi$.
In order to rule out these cases one should additionally require (at least in 2D) that for sub-body to contain no defects, each generator of the fundamental group of the body has to have turning number $1$ (see, e.g., \cite[Theorem~5]{KMS15}, \cite[Definition~3.1]{KM23}).
However, these subtleties are not essential to the rest of this paper, for which the simpler Definition~\ref{def:defects} is sufficient.
\end{remark}

We now show that the disclination content always lies in the symmetry group of the body:
\begin{proposition}\label{prop:disclination_and_symmetry}
Let $(\M,\{U_\alpha, \P_\alpha\}_{\alpha},\calW)$ be a body with discrete defects, and let $\gamma:[0,1]\to \M$ be a closed curve based at $p$, with $p\in U_\alpha$ for some $\alpha$.
Then its disclination content $c_\gamma$ satisfies
\[
\P_\alpha(p) \circ c_\gamma \circ \P_\alpha^{-1}(p) \in \mathcal{G},
\]
where $\mathcal{G}\subset \SO(d)$ is the symmetry group of $\calW$.
\end{proposition}

\begin{proof}
Write $\gamma$ as $\gamma = \gamma_n *\ldots *\gamma_1$, where $\gamma_i \subset U_{\alpha_i}$, starting at a point $p_i$ and ending at a point $p_{i+1}$, where $p_1 = p_{n+1} = p$.
By the definition of the connection associated with the lattice structure, we have that
\[
c_\gamma = \P_{\alpha_n}^{-1}(p_{n+1})\circ \P_{\alpha_n}(p_{n}) \circ \ldots \circ  \P_{\alpha_2}^{-1}(p_3)\circ \P_{\alpha_2}(p_2) \circ  \P_{\alpha_1}^{-1}(p_2)\circ \P_{\alpha_1}(p_1) ,
\]
and thus
\[
\P_\alpha(p) \circ c_\gamma \circ \P_\alpha^{-1}(p) = \brk{\P_{\alpha}(p)\circ \P_{\alpha_n}^{-1}(p)}\circ \ldots \circ  \brk{\P_{\alpha_2}(p_2) \circ  \P_{\alpha_1}^{-1}(p_2)}\circ \brk{\P_{\alpha_1}(p) \circ \P_\alpha^{-1} (p)}.
\]
By \eqref{eq:symmetry_group}, all the terms in brackets in the righthand side belong to $\mathcal{G}$, which completes the proof.
\end{proof}

\begin{example}\label{ex:disclination}[A single disclination]
A two-dimensional body with one disclination (of positive curvature) can be obtained by removing a sector from a two-dimensional annulus, and gluing the opposite edges of the sector.
Start with a domain 
\[
\Omega = \{(r\cos \theta, r\sin \theta) ~:~ r \in (r_0,r_1),\, \theta \in (0,2\pi\alpha)\}\subset \R^2
\]
for some $r_1>r_0\ge 0$, and $\alpha \in (0,1)$.
Endow this body with the trivial implant map $\P_0=\id$.
Now, glue the edges $\theta =0$ and $\theta = 2\pi \alpha$ by considering the map 
\[
\chi(r\cos\theta,r\sin\theta) = \brk{r\cos \frac{\theta}{\alpha},r\sin \frac{\theta}{\alpha}}.
\]
Its image is the smooth body 
\[
\M=\{(r\cos \vp, r\sin \vp) ~:~ r \in (r_0,r_1),\, \vp \in \mathbb{S}^1\},
\]
and the image of the reference map is
\[
\begin{split}
\P_1 &= d(\chi^{-1}) 
	= e_1 \otimes \brk{\cos(\alpha\vp)\,dr - \alpha r\sin(\alpha\vp)\,d\vp} +  e_2 \otimes \brk{\sin(\alpha\vp)\,dr + \alpha r\cos(\alpha\vp)\,d\vp} \\
\end{split}
\]
where $(e_1,e_2)$ is the standard orthonormal frame in $\R^2$.
This map is defined as a smooth reference map on $U_1 = \{(r\cos \vp, r\sin \vp) ~:~ r \in (r_0,r_1),\, \vp \in (0,2\pi)\}$, but not on $\M$.
To complete the lattice structure, we can similarly define a map $\P_2$ using the same expression as $\P_1$,
but define it on the reference neighborhood $U_2 = \{(r\cos \vp, r\sin \vp) ~:~ r \in (r_0,r_1),\, \vp \in (-\pi,\pi)\}$.
Both $\P_1$ and $\P_2$ are exact, and in particular, closed.
Now, to complete the definition of the elastic body we need to specify the archetype $\calW$.
Since we assume that the body is smooth, we need the map 
\[
W(\cdot) = \calW(\cdot \circ \P_1^{-1})
\] 
to be smooth also on the ray $\{\vp = 0\}$, which reduces to requiring 
\[
\P_1|_{\vp = 0^-} \circ  \P_1^{-1}|_{\vp = 0^+}\in \mathcal{G}.
\]
Since 
\[
\begin{split}
\P_1|_{\vp=0^+} &= e_1 \otimes dr + e_2 \otimes \alpha r \,d\vp = e_1 \otimes dx + \alpha e_2\otimes dy = \brk{\begin{matrix} 1 & 0 \\ 0 & \alpha\end{matrix}}\\
\P_1|_{\vp=0^-} &= e_1 \otimes \brk{\cos(2\pi \alpha)\,dr - \alpha r\sin(2\pi\alpha)\,d\vp} +  e_2 \otimes \brk{\sin(2\pi\alpha)\,dr + \alpha r\cos(2\pi\alpha)\,d\vp} \\
	&= \brk{\begin{matrix} \cos(2\pi \alpha) &-\alpha \sin(2\pi\alpha) \\ \sin(2\pi\alpha) & \alpha \cos(2\pi\alpha)\end{matrix}},
\end{split}
\]
we must have 
\[
\P_1|_{\vp = 0^-} \circ  \P_1^{-1}|_{\vp = 0^+} 
	= \brk{\begin{matrix} \cos(2\pi \alpha) &  -\sin(2\pi\alpha) \\ \sin(2\pi\alpha) &  \cos(2\pi\alpha)\end{matrix}}\in \mathcal{G}.
\]
That is, $\mathcal{G}$ contains a rotation by $2\pi \alpha$.
This same condition also arises by using $\P_2$ instead of $\P_1$. 
Since the holonomy of the Levi-Civita connection in this case is also a rotation by $-2\pi\alpha$, this is just a manifestation of the result of Proposition~\ref{prop:disclination_and_symmetry}, deduced directly from the requirement that the body will be a smooth body. 
In other words, not only the gluing of $\Omega$ via $\chi$ should be smooth, but the gluing of the triplet $(\Omega, \{\Omega, \P_0\},\calW)$.
Now, for an isotropic archetype (Example~\ref{ex:archetypes}(1)) this poses no restriction on $\alpha$, however the archetype has a discrete hexagonal symmetry as in Example~\ref{ex:archetypes}(2), it restricts $\alpha$ to be a multiple of $\frac{1}{6}$.
\end{example}

\begin{example}\label{ex:dislocation}[A single dislocation \cite{KM23}]
A two-dimensional body with a single edge-dislocation, whose Burgers vector is of magnitude $\e$ can be obtained as $(\M,\{\M,\P\},\calW)$ where
\[
\M=\{(r\cos \vp, r\sin \vp) ~:~ r \in (\e,r_1),\, \vp \in [0,2\pi]\},
\]
and
\[
\P = \id + \frac{\e}{2\pi}e_1 \otimes d\vp.
\]
It is obvious that $\P$ is closed (yet not exact), hence this body is a body with discrete defects and no disclinations.
Since, for any $r_0\in (\e,r_1)$, we have
\[
\int_{r=r_0} \P = \e e_1,
\]
the Burgers vector of this body is the vector field $\e\P^{-1}(e_1)$, meaning that for any simply connected, positively-oriented curve $\gamma$ based at $p$, its Burgers vector is $b_\gamma = \e(\P(p))^{-1}(e_1)$.
Note that, unlike the previous example, this construction is consistent with any archetype $\calW$, regardless of its symmetry group.
The geometry of this example is the same as the traditional Volterra constructions edge dislocations, obtained either by removing of a strip from a two dimensional body and gluing its edges, or by cutting the body along a ray and gluing it after translating one side. This uniqueness of the two-dimensional dislocation geometry is shown in \cite[Theorem~3.3]{KM23}.
\end{example}

\section{Discussion --- homogenization}\label{sec:discussion}
The description of bodies with discrete defects is rather unambiguous, ever since Volterra \cite{Vol07} classified different types of disclinations and dislocations (though these terms were coined later).
Their magnitudes are typically obtained by the holonomy and Burgers vector procedures described in Definition~\ref{def:defects_content}, or equivalent ways (e.g., their discrete counterparts on a defected lattice).
The infinitesimal versions of these procedures, when applied to a general affine connection $\nabla$ on $\M$, are the Riemann curvature tensor and the torsion tensor of $\nabla$ (see \cite[Chapter III, \S2, \S4]{Sch54} for precise statements and proofs, also \cite[p.~366]{EKM20}).
This observation led Bilby et al.~\cite{BBS55} to model a continuous distribution of dislocation as an affine connection with torsion; Wang \cite{Wan67} then related this model to the mechanical behavior via the notion of material connection, as described above.
The association of the curvature tensor of the connection as a continuous distribution of disclination appeared later, in linear \cite{DW73a,DW73b,Kro81} and nonlinear \cite{MR02,RG17} geometric theories.

While these observations are suggestive and elegant, they are not a rigorous derivation of a continuous distribution of defects model from the discrete one; this derivation is still, to a large extent, missing.
For dislocations (in 2D) this homogenization was obtained in a series of works culminating in \cite{EKM20}. 
Their results can be summarized as follows:
 \begin{enumerate}
 \item It is possible to obtain any smooth two-dimensional body with dislocations $(\M,\g,\nabla)$, where $\nabla$ is a flat, metric connection with torsion, as a limit of bodies $(\M^n, \P^n)$ with finitely many dislocations in the sense of Definition~\ref{def:defects}.
 The limit is obtained as the dislocations get smaller and denser, and $\nabla$ is the limit of the Levi-Civita (material) connections associated with $\P_n$.
 \item Under some technical assumptions on an archetype $\calW$, the energy $E_{(\M^n,W^n,\Vol^n)}$ associated with $(\M^n, \P^n,\calW)$ converges, in a sense of $\Gamma$-convergence, to an energy $E_{(\M,W,\Vol)}$ with the same archetype $\calW$. 
 The limit connection $\nabla$ is a material connection of the limit energy functional.
 \item If the archetype is isotropic, the limit material connection $\nabla$ is not unique, and one cannot obtain from $E_{(\M,W,\Vol)}$ the connection $\nabla$ or its torsion tensor (dislocation density).
 In fact, $E_{(\M,W,\Vol)}$ depends only on the archetype $\calW$ and the metric $\g$.
 \end{enumerate}  
 
A homogenization result for disclinations should follow the same line: first, obtaining a body $(\M,\g,\nabla)$ with a non-flat $\nabla$ (whose curvature represents the disclination density according to the geometric paradigm), as a limit of bodies $(\M^n,\{U_\alpha^n, \P_\alpha^n\}_{\alpha},\calW)$ with discrete defects. 
Then, obtain an energy $E_{(\M,W,\Vol)}$ as a limit of $E_{(\M^n,W^n,\Vol^n)}$, where $E_{(\M,W,\Vol)}$ has an archetype $\calW$ and a material connection $\nabla$.
 The discussion below will show that this program has serious limitations:
 \begin{claim}\label{claim:main}
 \begin{enumerate}
 \item If $\calW$ has discrete symmetry group $\mathcal{G}$ (crystalline materials), then one cannot obtain a non-flat $\nabla$ as a limit of the Levi-Civita connections of bodies with discrete defects $(\M^n,\{U_\alpha^n, \P_\alpha^n\}_{\alpha},\calW)$.
 	This is in sharp contrast to dislocations, where the most complete homogenization result is obtained in this case.
 \item If $\calW$ is isotropic, then it is possible to obtain a non-flat material connection $\nabla$ of a body $(\M,\{U_\alpha, \P_\alpha\}_{\alpha},\calW)$, as a limit of the Levi-Civita connections of bodies with discrete defects $(\M^n,\{U_\alpha^n, \P_\alpha^n\}_{\alpha},\calW)$; it is further possible that the associated energies $E_{(\M^n,W^n,\Vol^n)}$ $\Gamma$-converge to $E_{(\M,W,\Vol)}$.
 However, one cannot obtain from the limit energy $E_{(\M,W,\Vol)}$ the connection $\nabla$, its curvature tensor, or whether $E_{(\M,W,\Vol)}$ was obtained as homogenization limit of disclinations, dislocations, or both.
 \end{enumerate}
 \end{claim}
 
We start with the case of discrete symmetry group.
The following is a corollary of Proposition~\ref{prop:disclination_and_symmetry}:
\begin{corollary} 
Let $\calW$ be a solid, undistorted archetype with a discrete symmetry group $\mathcal{G}$.
\begin{enumerate}
\item 
The dislocation content of any curve is any body with discrete defects $(\M,\{U_\alpha, \P_\alpha\}_{\alpha},\calW)$ is either zero or is bounded away from the identity.
\item Let $\M$ be a manifold, and $(\M^n,\{U_\alpha^n, \P_\alpha^n\}_{\alpha},\calW)$ be a sequence of bodies with disclinations such that $\M^n\subset \M$. 
Then the Levi-Civita connections $\nabla^n$ associated with $\M^n$ cannot converge to a non-flat, smooth connection $\nabla$ on $\M$, for any notion of convergence that implies convergence of the parallel transport (i.e., of the lattice directions).
\end{enumerate}
\end{corollary}

\begin{proof}
The first part follows immediately from the proposition, as a discrete $\mathcal{G}$ implies that for any two distinct elements $g_1,g_2\in \mathcal{G}$, $|g_1-g_2|>c$ for some $c>0$.
The second part follows from the first, as a smooth connection with non-zero curvature tensor at $p\in \M$ must have a nontrivial holonomy, yet arbitrarily close to $\id_{T_p\M}$, on sufficiently short closed curves based at $p$.
\end{proof}

Thus, we cannot hope for a homogenization theory of continuous distribution of disclinations for bodies with discrete symmetry groups (e.g., crystalline materials), under the assumption that their geometric structures corresponds to their mechanical behavior.
This concludes the first part of Claim~\ref{claim:main}.

We now describe, without getting into technical details, how the second part can be obtained.
Fix an isotropic archetype $\calW$.
Consider a body $\M\subset \R^2$, with a smooth Riemannian metric $G$ on it, and let $\nabla$ be its Levi-Civita connection.
Now construct a sequence of piecewise flat Riemannian manifolds as follows (see also \cite[\S2.3.2, Example 3]{KM16a}):
Triangulate $(\M,\g)$ with geodesics of length of order $1/n$, such that the angle in each geodesic triangle is uniformly bounded away from $0$ or $\pi$. 
Replace each triangle with a Euclidean triangle of the same edge lengths, and remove the vertices. 
The result is a smooth Riemannian manifold $(\M^n,G^n)$, which is locally flat and whose Levi-Civita connection $\nabla^n$ has holonomy (i.e., disclination content) around the removed vertices.
In the vicinity of each vertex $(\M^n,G^n)$ is isometric to a body with a single disclination as described Example~\ref{ex:disclination}. Thus, one can pullback the implant maps from there to these sub-bodies.
Since $\calW$ is isotropic, these implant maps of the neighborhoods of all the vertices are compatible with each other, and together they form elastic bodies with discrete defects $(\M^n,\{U_\alpha^n, \P_\alpha^n\}_{\alpha},\calW)$ whose intrinsic metrics are $\g^n$ and their Levi-Civita (material) connections are $\nabla^n$.
We can easily embed $\M^n$ into $\M$ almost isometrically, by mapping the corresponding geodesic triangles.
Thus, we can assume that $\M^n\subset \M$. 
By construction, for any curve $\gamma$ that lies in $\bigcap_n \M^n$, the parallel transport along $\gamma$ with respect to $\nabla^n$ converges to the parallel transport with respect to $\gamma$, and furthermore, the metrics $\g^n$ converge uniformly to $\g$ \cite{KM16a}.
Thus, it follows from the $\Gamma$-convergence results of \cite{KM16a,EKM20} that $E_{(\M^n,W^n,\Vol^n)}$ $\Gamma$-converge to $E_{(\M,W,\Vol)}$, provided that $\calW$ satisfies some technical assumptions (namely quasiconvexity and standard growth conditions).

However, due to the non-uniqueness of material connections in isotropic materials, the same energy $E_{(\M,W,\Vol)}$ can arise as a homogenization limit of bodies with discrete defects that contain no disclinations. 
In such cases, the material connections $\tilde{\nabla}^n$ may converge to a different, curvature-free material connection $\tilde{\nabla}$ (as demonstrated in the results of \cite{EKM20} discussed earlier). 
Consequently, if we only have access to the mechanical behavior---the energy functional $E_{(\M,W,\Vol)}$---of an isotropic elastic body $(\M,W,\Vol)$, \emph{it is impossible to distinguish whether the observed behavior originates from the homogenization of bodies containing disclinations, dislocations, or both}. 
In particular, the curvature and torsion tensors of the limiting material connection do not manifest in the energy functional. 
This concludes the proof of Claim~\ref{claim:main}.

In summary, we have demonstrated that, unlike dislocations, obtaining models of continuous distribution of disclinations is limited to materials with a continuous isotropic group. In such cases, it is impossible to discern from the mechanical behavior whether the response arises from distributed disclinations or to dislocations.
Thus, from this perspective, the mechanical interpretation of the curvature tensor of a connection, as a descriptor of continuous distribution of disclinations, is unclear.

\begin{remark}
This paper discusses disclinations and dislocations, leaving point-defects of various kinds outside its scope.
A rigorous derivation of homogenization limits of these defects remains an open problem (in a geometric paradigm, an attempt towards it appears in \cite{KMR18}).
In fact, even the identification of continuous distribution of point-defects as non-metricity of a connection is less clear cut than the identification of dislocations and disclinations with torsion and curvature \cite[p.~304]{Kro81}.
Moreover, unlike torsion and curvature, the non-metricity tensor does not align with the constitutive paradigm of Noll and Wang.
Indeed, in the framework described in \S2, material connections are always metrically-consistent with the intrinsic metric \cite[Proposition~2]{EKM20}.
As a result, the very framework within which the geometry and mechanical behavior of bodies with continuously distributed point-defects---arising as limits of bodies with finitely many defects---can be rigorously described, remains unclear.
\end{remark}


{\footnotesize
\providecommand{\bysame}{\leavevmode\hbox to3em{\hrulefill}\thinspace}
\providecommand{\MR}{\relax\ifhmode\unskip\space\fi MR }
\providecommand{\MRhref}[2]{%
  \href{http://www.ams.org/mathscinet-getitem?mr=#1}{#2}
}
\providecommand{\href}[2]{#2}

}


\begin{thebibliography}{deW73b}

\bibitem[BBS55]{BBS55}
B.A. Bilby, R.~Bullough, and E.~Smith, \emph{Continuous distributions of
  dislocations: A new application of the methods of {Non-Riemannian} geometry},
  Proc. Roy. Soc. A \textbf{231} (1955), 263--273.

\bibitem[CvM21]{CvM21}
P. Cesana and P. van Meurs, \emph{Discrete-to-continuum limits of
  planar disclinations}, ESAIM: Cont. Opt. Calc. Var. \textbf{27} (2021), 23.

\bibitem[deW73a]{DW73a}
R. deWit, \emph{Theory of disclinations: {II.} continuous and discrete
  disclinations in anisotropic elasticity}, Journal of Research of the National
  Bureau of Standards. Section A, Physics and Chemistry \textbf{77} (1973),
  no.~1, 49.

\bibitem[deW73b]{DW73b}
\bysame, \emph{Theory of disclinations: {III.} continuous and discrete
  disclinations in isotropic elasticity}, Journal of Research of the National
  Bureau of Standards. Section A, Physics and Chemistry \textbf{77} (1973),
  no.~3, 359.

\bibitem[EKM20]{EKM20}
M.~Epstein, R.~Kupferman, and C.~Maor, \emph{Limits of distributed dislocations
  in geometric and constitutive paradigms}, Geometric Continuum Mechanics
  (R.~Segev and M.~Epstein, eds.), Birkh\"auser Basel, 2020.
  
 
\bibitem[GY18]{GY18}
A. Golgoon and A. Yavari, \emph{Line and point defects in nonlinear anisotropic solids}, Zeitschrift f{\"u}r angewandte Mathematik und Physik, \textbf{69} (2018), 1--28.


\bibitem[KM15]{KM15}
R. Kupferman and C.~Maor, \emph{The emergence of torsion in the continuum
  limit of distributed edge-dislocations}, Journal of Geometric Mechanics
  \textbf{7} (2015), no.~3, 361--387.

\bibitem[KM16a]{KM16a}
\bysame, \emph{Limits of elastic models of converging {Riemannian} manifolds},
  Calc. Variations and PDEs \textbf{55} (2016), 40,
  http://link.springer.com/article/10.1007/s00526-016-0979-6.

\bibitem[KM16b]{KM16}
\bysame, \emph{{Riemannian} surfaces with torsion as homogenization limits of
  locally-{Euclidean} surfaces with dislocation-type singularities}, Proc. Roy.
  Soc. Edin. A \textbf{146} (2016), no.~4, 741--768.

\bibitem[KM23]{KM23}
\bysame, \emph{From {Volterra} dislocations to
  strain-gradient plasticity}, Calc. Var. PDE, to appear.
  \url{https://arxiv.org/abs/2302.02368}

\bibitem[KMR18]{KMR18}
R. Kupferman, C.~Maor, and R. Rosenthal, \emph{Non-metricity in the continuum
  limit of randomly-distributed point defects}, Israel Journal of Mathematics
  \textbf{223} (2018), 75--139.

\bibitem[KMS15]{KMS15}
R. Kupferman, M. Moshe, and J.~P. Solomon, \emph{Metric description of
  singular defects in isotropic materials}, Arch. Rat. Mech. Anal. \textbf{216}
  (2015), 1009--1047.

\bibitem[KO20]{KO20}
R.~Kupferman and E.~Olami, \emph{Homogenization of edge-dislocations as a weak
  limit of {de-Rham} currents}, Geometric {C}ontinuum {M}echanics (R.~Segev and
  M.~Epstein, eds.), Birkh\"auser Basel, 2020.

\bibitem[Kon55]{Kon55}
K.~Kondo, \emph{Geometry of elastic deformation and incompatibility}, Memoirs
  of the Unifying Study of the Basic Problems in Engineering Science by Means
  of Geometry (K.~Kondo, ed.), vol.~1, 1955, pp.~5--17.

\bibitem[Kr{\"o}81]{Kro81}
E.~Kr{\"o}ner, \emph{Contiuum theory of defects}, Physics of Defects -- Les
  Houches Summer School Proceedings (Amsterdam) (R.~Balian, M.~Kleman, and
  J.-P. Poirier, eds.), North-Holland, 1981.

\bibitem[MR02]{MR02}
M.F. Miri and N. Rivier, \emph{Continuum elasticity with topological
  defects, including dislocations and extra-matter}, J. Phys. A: Math. Gen.
  \textbf{35} (2002), 1727--1739.

\bibitem[Nol59]{Nol58}
W. Noll, \emph{A mathematical theory of the mechanical behavior of
  continuous media}, Arch. Rat. Mech. Anal. \textbf{2} (1958/59), no.~1,
  197--226.

\bibitem[Nye53]{Nye53}
J.~F. Nye, \emph{Some geometrical relations in dislocated crystals}, Acta
  Metallurgica \textbf{1} (1953), 153--162.

\bibitem[RG17]{RG17}
A. Roychowdhury and A. Gupta, \emph{Non-metric connection and metric
  anomalies in materially uniform elastic solids}, Journal of Elasticity
  \textbf{126}(1) (2017), 1--26.

\bibitem[SY17]{SY17}
S. Sadik and A. Yavari, \emph{On the origins of the idea of the multiplicative decomposition of the deformation gradient}, Math. Mech. Solids, \textbf{22}(4) (2017), 771--772.

\bibitem[Sch54]{Sch54}
J.A. Schouten, \emph{Ricci-calculus}, Springer-Verlag Berlin Heidelberg, 1954.


\bibitem[Vol07]{Vol07}
V.~Volterra, \emph{Sur l'{\'e}quilibre des corps {\'e}lastiques multiplement
  connexesquilibre des corps {\'e}lastiques multiplement connexes}, Ann. Sci.
  Ecole Norm. Sup. Paris \textbf{24} (1907), 401--518.

\bibitem[Wan67]{Wan67}
C.~C. Wang, \emph{On the geometric structures of simple bodies, a mathematical
  foundation for the theory of continuous distributions of dislocations}, Arch.
  Rat. Mech. Anal. \textbf{27} (1967), no.~1, 33--94.

\bibitem[YG12a]{YG12}
A.~Yavari and A.~Goriely, \emph{Weyl geometry and the nonlinear mechanics of
  distributed point defects}, Proceedings of the Royal Society A \textbf{468}
  (2012), 3902--3922.

\bibitem[YG12b]{YG12b}
\bysame, \emph{{Riemann-Cartan} geometry of nonlinear
  dislocation mechanics}, Arch. Rat. Mech. Anal. \textbf{205} (2012), no.~1,
  59--118.

\bibitem[YG13]{YG13}
\bysame, \emph{Riemann--{C}artan geometry of nonlinear
  disclination mechanics}, Mathematics and Mechanics of Solids \textbf{18}
  (2013), no.~1, 91--102.

\end{thebibliography}
\end{document}